\documentclass[a4paper,11pt]{article}
\usepackage{fullpage}
\usepackage{amsmath,amsthm,amssymb,amsfonts,mathrsfs,amsthm}
\usepackage{dsfont}
\usepackage{graphicx,hyperref}
\usepackage{enumitem}

\def\Tr{\operatorname{Tr}}
\def\>{\rangle}
\def\<{\langle}
\def\N#1{\left|\!\left|{#1}\right|\!\right|}

\def\mE{\mathcal{E}}

\def\mN{\mathcal{N}}

\def\sH{\mathcal{H}}

\def\openone{\mathds{1}}
\newcommand{\set}[1]{\mathcal{#1}}
\newcommand{\op}[1]{\mathsf{#1}}

\newcommand{\pguess}{P_{\operatorname{guess}}}

\newcommand{\defeq}{\stackrel{\textup{\tiny def}}{=}}

\newcommand{\Hmin}{H_{\operatorname{min}}}
\newcommand{\succth}{\succeq_\omega}

\renewcommand{\qedsymbol}{\nobreak \ifvmode \relax \else
  \ifdim \lastskip<1.5em \hskip-\lastskip \hskip1.5em plus0em
  minus0.5em \fi \nobreak \vrule height0.75em width0.5em
  depth0.25em\fi}

\renewcommand{\ge}{\geqslant}
\renewcommand{\le}{\leqslant}

\newtheorem{theorem}{Theorem}

\newtheorem{lemma}{Lemma}
\newtheorem{proposition}{Proposition}

\newtheorem{definition}{Definition}

\theoremstyle{remark}

\theoremstyle{definition}

\newcommand{\bea}{\begin{eqnarray}}
\newcommand{\eea}{\end{eqnarray}}
\newcommand{\be}{\begin{equation}}
\newcommand{\ee}{\end{equation}}

\begin{document}


\title{Fully quantum second-law--like statements\\from the theory of statistical comparisons}

\author{{\small \sc Francesco Buscemi}\footnote{buscemi@is.nagoya-u.ac.jp}\\
    \footnotesize{\em Graduate School of Information Science, Nagoya University}\\
    \footnotesize{\em Furo-cho, Chikusa-ku, Nagoya, 464-8601 Japan}}

\date{\today}

\maketitle

\begin{abstract}
	In generalized resource theories, one aims to reformulate the problem of deciding whether a suitable transition (typically a transition that preserves the Gibbs state of the theory) between two given states $\rho$ and $\sigma$ exists or not, into the problem of checking whether a set of second-law--like inequalities hold or not. The aim of these preliminary notes is to show how the theory of statistical comparisons (in the sense of Blackwell, LeCam, and Torgersen) can be useful in such scenarios. In particular, we propose one construction, in which the second laws are formulated in terms of a suitable conditional min-entropy. Though a general, fully quantum result is also presented, stronger results are obtained for the case of qubits, and the case of $\sigma$ commuting with the Gibbs state.
\end{abstract}


Recently, a great deal of attention has been devoted to the reconstruction of thermodynamics from operational/information-theoretic first principles~\cite{rio_thermodynamic_2011,faist_quantitative_2012,aberg_truly_2013,brandao_resource_2013,gour_resource_2013,horodecki_fundamental_2013,aberg_catalytic_2014,cwiklinski_towards_2014,frenzel_reexamination_2014,reeb_improved_2014,narasimhachar_low-temperature_2014,skrzypczyk_work_2014,alhambra_what_2015,brandao_second_2015,lostaglio_quantum_2015,lostaglio_description_2015,weilenmann_axiomatic_2015,gallego_defining_2015}, within the context of `resource theories'~\cite{brandao_resource_2013,horodecki_quantumness_2012,coecke_mathematical_2014,brandao_general_2015}. The benefits of such an approach are apparent: one the one hand, the mathematical foundations of thermodynamics are simplified, so that older results can be generalized and new ideas can be developed in a unified rigorous framework. On the other hand, the operational ground is finally cleared from old preconceptions that, stemming from a `classical' intuition, hindered for too long the development of a consistent, fully quantum theory of thermodynamics.

Particular merit, towards the resurgence of general interest in this direction, surely goes to the work of Lieb and Yngvason~\cite{lieb_physics_1999,lieb_entropy_2013,lieb_entropy_2014}, in which it is shown that, under natural assumptions, an adiabatic process connecting two given states exists, if and only if the entropy of the initial state is no less than the entropy of the final one~\cite{lieb_physics_1999}. In other words, adiabatic processes, in the Lieb-Yngvason framework, are completely characterized as `entropy-non-increasing transitions' from one state to another, perfectly in line with an `order theoretical' approach.

The general viewpoint advocated by Lieb and Yngvason is that of looking upon thermodynamics as `just another' utility theory, formally similar, to some extent, to game theory, financial mathematics, and, ultimately, statistical decision theory. Indeed, the fundamental problem, in thermodynamics as in other utility theories, is that of giving necessary and sufficient conditions for the existence of a `transition' (whose properties are defined by the theory itself) between given families of states, in terms of a set of operationally well-motivated utility functions.

The aim of these preliminary notes is to argue that, for this kind of problems, a suitable framework is provided by the theory of statistical comparisons, laid down in the classical case by Blackwell~\cite{blackwell_equivalent_1953}, LeCam~\cite{cam1964sufficiency,cam_asymptotic_1986} and Torgersen~\cite{torgersen1970comparison,torgersen_comparison_1991}, and recently generalized to the quantum case~\cite{buscemi2012comparison,matsumoto2010quantum,jencova_comparison_2012}. Such a framework has already been applied to various problems in quantum information theory~\cite{buscemi2012all,buscemi_game-theoretic_2014,jencova_randomization_2014}, quantum estimation theory~\cite{matsumoto2012input}, quantum measurement theory~\cite{buscemi2005clean,matsumoto_convertibility_2014}, and the theory of quantum Markov processes~\cite{matsumoto_loss_2012,buscemi_complete_2014,buscemi_equivalence_2014}. It is the aim of these notes to show how the theory of statistical comparisons may be useful in quantum thermodynamics (and other resource theories) too.

\section{Preliminary remarks}

In what follows, as an example, we consider the case of a resource theory with one Gibbs state $\omega$: such state could be the thermal state corresponding to some preferred Hamiltonian at some given temperature (i.e., the case of athermality resource theory), or could be the state invariant under the action of a group of transformations (i.e., the case of asymmetry resource theory).

The goal we are aiming at is to show that the theory of comparisons of statistical models provides a good framework to study statements like the following:
\begin{quote}
	\textbf{A second-law--like statement}. An `$\omega$-preserving transition' from $\rho$ to $\sigma$ is possible if and only if $\sigma$ is `closer' than $\rho$ to $\omega$, i.e., the Gibbs state of the theory.
\end{quote}
In the above statement, the term `closer' should be expressed in terms of a set of inequalities, formalizing the idea that $\sigma$ is no more `resourceful' than $\rho$, in analogy with the conventional interpretation of the second law of thermodynamics as a decrease of free energy.

A very closely related question we are going to answer is the following:
\begin{quote}
	\textbf{Converse to the resource-processing inequality}. What are the most general transitions that never increase the `operational distances' with the fixed Gibbs state of the theory?
\end{quote}
As before, the term `operational distance' needs to be defined: we will do this in what follows.

\subsection{Notation}

In what follows, all sets are finite and Hilbert spaces are finite-dimensional.
\begin{itemize}
	\item Sets are denoted by $\set{X}=\{x:x\in\set{X}\}$, $\set{Y}=\{y:y\in\set{Y}\}$, etc.
	\item A probability distribution over $\set{X}$ is a function $p:\set{X}\to[0,1]$ such that $\sum_xp(x)=1$.
	\item The set of all probability distributions over $\set{X}$ is denoted by $\op{P}(\set{X})$.
	\item Quantum systems are labeled by capital letters $Q$, $R$, etc, and the associated Hilbert spaces are denoted by $\sH_Q$, $\sH_R$, etc.
	\item The set of linear operators acting on a Hilbert space $\sH$ is denoted by $\op{L}(\sH)$.
	\item States of $Q$ are represented by density matrices, i.e., operators $\rho\in\op{L}(\sH)$ such that $\rho\ge0$ and $\Tr[\rho]=1$.
	\item The set of density matrices acting on a Hilbert space $\sH$ is denoted by $\op{S}(\sH)$.
	\item An ensemble $\mE$ is given by a probability distribution $p\in\op{P}(\set{X})$ and a family of density matrices $\{\rho^x:x\in\set{X}\}$. It models the situation in which the uncertainty about the state of a quantum system is purely classical, i.e., the system is known to be in one of the states $\rho^x$ with probability $p(x)$.
	\item A positive-operator valued measure (POVM) is a function $P:\set{X}\to\op{L}(\sH)$ such that $P(x)\ge0$ and $\sum_xP(x)=\openone$. For the sake of readability, we will often write the argument $x$ as a superscript, i.e., $P^x$ rather than $P(x)$.
	\item The set of POVMs from $\set{X}$ to $\op{L}(\sH)$ is denoted by $\op{M}(\set{X},\sH)$.
	\item Transitions (mappings) between states of two (generally different) quantum systems $Q$ and $Q'$ are restricted in this paper to completely positive trace-preserving (CPTP) linear maps $\Phi:\op{L}(\sH_Q)\to\op{L}(\sH_{Q'})$. We will often use the word `quantum channel' as a synonym of `CPTP linear map.'
\end{itemize}

\section{Statistical comparisons of pairs of quantum states}

Suppose we are given two pairs of probability distributions $(p_0,p_1)$ and $(q_0,q_1)$, or two pairs of density matrices $(\rho_0,\rho_1)$ and $(\sigma_0,\sigma_1)$. What pair, in each case, is `more informative?' Surely, the answer depends on the task at hand, i.e., how is the `information content' defined and measured.

To answer this sort of questions constitutes one of the main the motivations for a vast body of works, collectively referred to as the theory of comparisons of statistical models~\cite{torgersen_comparison_1991, cohen_comparisons_1998, liese_statistical_2008}. The main idea, roughly summarized, is to show that the existence of a suitable transition $p_i\mapsto q_i$ ($\rho_i\mapsto\sigma_i$), for $i=0,1$, is equivalent to the existence of a family of functions $\{f_t\}_t$ such that $f_t(p_0,p_1)\ge f_t(q_0,q_1)$ ($f_t(\rho_0,\rho_1)\ge f_t(\rho_0,\rho_1)$) for all $t$.

The point is that one typically tries to base the comparison on functions $f_t$ that enjoy a direct operational interpretation, so that they can be seen as \textit{utility functions}. While in some cases it might be that just one function is enough to completely answer the question about the existence of a transition, most results involve a parameter $t$ that is continuous, so that the functions to be considered are actually uncountably many. This however does not make the results so obtained any less striking or less powerful: in fact, deep theories (above all, LeCam's asymptotic theory of estimation~\cite{cam_asymptotic_1986}) have been constructed starting from these ideas.

It should be clear now how the theory of statistical comparison is apt to investigate second-law--like statements as that presented above: if an `$\omega$-preserving transition' is to be found between $\rho$ and $\sigma$, we can consider the pairs $(\rho,\omega)$ and $(\sigma,\omega)$. Therefore, the problem of formulating second-law--like statements for the existence of an $\omega$-preserving transition between \textit{two states}, can be seen as a special case of statistical comparisons between \textit{two pairs} of states.

Intentionally leaving aside the enormous amount of results valid when comparing families of probability distributions (referred to as \textit{statistical models} or \textit{statistical experiments}, see for example the monumental work by Torgersen~\cite{torgersen_comparison_1991}), in these notes we will focus on the semiquantum and fully quantum cases.

\section{The information-bearing ordering}\label{sec:info-ordering}

Let us consider the following situation. Imagine that a quantum system $Q$ can be prepared either in state $\rho_0$ or in state $\rho_1$, while another quantum system $Q'$ can be prepared either in state $\sigma_0$ or in state $\sigma_1$. Which setup is preferable from an information-theoretic viewpoint? Equivalently: given two pairs of density matrices, $(\rho_0,\rho_1)$ in $\op{S}(\sH_Q)$ and $(\sigma_0,\sigma_1)$ in $\op{S}(\sH_{Q'})$, which pair is \textit{always} more informative?

Surely, if there exists a quantum channel $\Phi:\op{L}(\sH_Q)\to\op{L}(\sH_{Q'})$ such that
\begin{equation}\label{eq:channel}
\Phi(\rho_i)=\sigma_i,\qquad i=0,1,
\end{equation}
then it is clear that the pair $(\rho_0,\rho_1)$ is \textit{always} more informative than $(\sigma_0,\sigma_1)$, simply because \textit{anything} that can be done with $(\sigma_0,\sigma_1)$ can also be done with $(\rho_0,\rho_1)$ upon processing the states with the channel $\Phi$. However, in general, neither a channel from the $\rho_i$'s to the $\sigma_i$'s, nor a channel from the $\sigma_i$'s to the $\rho_i$'s exists, so we expect that the ordering `being always more informative than' be just a partial ordering. The question we want to answer is the following: is it possible to give a non-trivial, operationally motivated, mathematically sound definition of `being always more informative,' implying the existence of a channel from the more informative pair to the less informative one?

We resort to the following construction. Since the quantum system $Q$ can be prepared either in state $\rho_0$ or in state $\rho_1$, we imagine to use $Q$ as an information-bearing system, i.e., we construct a classical-to-quantum channel $\mN_\rho:\set{X}\to\op{S}(\sH_Q)$, with $\set{X}=\{0,1\}$, such the input signal $i$ corresponds to output $\rho_i$. As it is customary in information theory, the action of the channel $\mN_\rho$ is linearly extended to any probabilistic mixture of input symbols, so that an initial distribution $(p,1-p)\in\op{P}(\set{X})$ is mapped by $\mN_\rho$ to the convex mixture $p\rho_0+(1-p)\rho_1$.

Typically, the channel $\mN_\rho$ is used to transmit correlations, i.e., it is used on correlated inputs: for any other set $\set{U}$ and any given joint probability distribution $p\in\op{P}(\set{U}\times\set{X})$, the action of $\mN_\rho$ gives rise to the bipartite classical-quantum state
\[
\begin{split}
\rho_{UQ}&\defeq\sum_{u\in\set{U}}\sum_{x\in\set{X}}p(u,x)|u\>\<u|_U\otimes\rho^x_Q\\
&=\sum_{u\in\set{U}}p(u)|u\>\<u|_U\otimes\rho^u_Q,
\end{split}
\]
with $\rho^u_Q\defeq\sum_{x\in\set{X}}p(x|u)\rho^x_Q$. In the above equation, we used the conventional notation, in which the classical random variable $U$ is fictitiously `encoded' on the perfectly distinguishable pure states $|u\>\<u|_U$ of an auxiliary quantum system (also labeled by $U$), associated with the Hilbert space $\sH_U$.

We then look at the probability that an observer, having access to $Q$, can guess the correct value of $U$, i.e.,
\[
\pguess(U|Q)_\rho\defeq\max_{P\in\op{M}(\set{U},\sH_Q)}\sum_{u\in\set{U}}p(u)\Tr[\rho^u_Q\ P^u_Q].
\]
As shown in Ref.~\cite{konig2009operational}, the surprisal of the guessing probability is equal to the conditional min-entropy, i.e.,
\[
-\log_2 \pguess(U|Q)_\rho=\Hmin(U|Q)_\rho.
\]

Of course, the same construction can be repeated for the pair $(\sigma_0,\sigma_1)$: in this case, we obtain the channel $\mN_\sigma$ that, for any set $\set{U}$ and any joint probability distribution $p\in\op{P}(\set{U}\times\set{X})$, gives rise to the corresponding cq-state given by
\[
\sigma_{UQ'}\defeq\sum_{u\in\set{U}}\sum_{x\in\set{X}}p(u,x)|u\>\<u|_U\otimes\sigma^x_{Q'}
\]

We now introduce the following definition.

\begin{definition}[Information ordering]\label{def:information-ordering}
	Let $(\rho_0,\rho_1)$ and $(\sigma_0,\sigma_1)$ be two pairs of density matrices in $\op{S}(\sH_Q)$ and in $\op{S}(\sH_{Q'})$, respectively. We say that $(\rho_0,\rho_1)$ is \emph{`more informative'} than $(\sigma_0,\sigma_1)$, and write \[(\rho_0,\rho_1)\succeq (\sigma_0,\sigma_1),\] if and only if, for any set $\set{U}$ and any joint probability distribution $p\in\op{P}(\set{U}\times\set{X})$,
	\[
	\Hmin(U|Q)_\rho\le\Hmin(U|Q')_\sigma.
	\]
\end{definition}

The information ordering given above embodies a very strong notion of distinguishability, one that looks at how good a pair of states is for encoding information, \textit{for all possible such encodings}. In the particular case in which $\set{U}$ is binary, the condition $(\rho_0,\rho_1)\succeq (\sigma_0,\sigma_1)$ is essentially equivalent (see below the proof of Lemma~\ref{lem:qubit}) to saying that $\N{\pi_0\rho_0-\pi_1\rho_1}_1\ge\N{\pi_0\sigma_0-\pi_1\sigma_1}_1$, for all choices of prior probabilities $\pi_0=1-\pi_1$.

Suppose now that $\rho_1=\sigma_1=\omega$, where $\omega$ is the Gibbs state (i.e., the `resourceless' state) of some resource theory. We can then introduce the following definition.

\begin{definition}[Thermal ordering]\label{def:thermal-ordering}
	Let $\rho$ and $\sigma$ be two density matrices in $\op{S}(\sH)$, and let $\omega\in\op{S}(\sH)$ be the Gibbs (or \emph{thermal}) state of the theory. We say that $\rho$ is \emph{`less thermal'} than $\sigma$, and write
	\[
	\rho\succth\sigma,
	\]
	if and only if $(\rho,\omega)\succeq(\sigma,\omega)$.
\end{definition}

In other words, $\rho\succth\sigma$ if and only if $\rho$ is not `less distinguishable' from $\omega$ than $\sigma$.

\section{The case of qubits}

In the case in which both pairs consists of 2-by-2 density matrices, we have the following:

\begin{lemma}\label{lem:qubit}
	Suppose that $\rho_0,\rho_1,\sigma_0,\sigma_1\in\op{S}(\mathbb{C}^2)$. Then, $(\rho_0,\rho_1)\succeq (\sigma_0,\sigma_1)$ if and only if there exists a CPTP map $\Phi:\op{L}(\mathbb{C}^2)\to\op{L}(\mathbb{C}^2)$ such that $\Phi(\rho_i)=\sigma_i$, for $i=0,1$. In fact, the comparison $(\rho_0,\rho_1)\succeq (\sigma_0,\sigma_1)$ can be restricted to sets $\set{U}$ with only two elements.
\end{lemma}

The above proposition is a rather direct consequence of a result by Alberti and Uhlmann~\cite{alberti_problem_1980}:

\begin{theorem}[Alberti-Uhlmann, 1980]
	Given two pair of states $(\rho_0,\rho_1)$ and $(\sigma_0,\sigma_1)$, all in $\op{S}(\mathbb{C}^2)$, there exists a CPTP map $\Phi$ such that $\Phi(\rho_i)=\sigma_i$ if and only if
	\begin{equation}\label{eq:alberti}
	\N{\rho_0-t\rho_1}_1\ge \N{\sigma_0-t\sigma_1}_1,
	\end{equation}
	for all $t\in\mathbb{R}$.
\end{theorem}

(In fact, the original statement has $t\in(0,+\infty)$, but, since the inequality~\eqref{eq:alberti} automatically holds for $t\in(-\infty,0]$, we extend the range of $t$ for the sake of simplicity.)

\begin{proof}[Proof of Lemma~\ref{lem:qubit}]
We only prove the `only if' part (the `if' part is trivial). Restricting the attention to binary $\set{U}=\{a,b\}$, we see that $\pguess(U|Q)_\rho\ge\pguess(U|Q')_\sigma$ for all $p\in\op{P}(\set{U}\times\set{X})$, is equivalent to
\[
\begin{split}
&\N{\{p(a,0)-p(b,0)\}\rho^0_Q-\{p(b,1)-p(a,1)\}\rho^1_Q}_1\\
\ge&\N{\{p(a,0)-p(b,0)\}\sigma^0_{Q'}-\{p(b,1)-p(a,1)\}\sigma^1_{Q'}}_1,
\end{split}
\]
for all $p\in\op{P}(\set{U}\times\set{X})$. If $p(a,0)-p(b,0)=0$, then the above condition is automatically satisfied. We can therefore restrict our attention to the case $p(a,0)-p(b,0)=s\neq 0$. Using the fact that $\N{cX}_1=|c|\N{X}_1$, the above condition is equivalent to
\[
\begin{split}
&\N{\rho^0_Q-\frac{p(b,1)-p(a,1)}{s}\rho^1_Q}_1\\
\ge&\N{\sigma^0_{Q'}-\frac{p(b,1)-p(a,1)}{s}\sigma^1_{Q'}}_1,
\end{split}
\]
for all $p\in\op{P}(\set{U}\times\set{X})$. By varying $p$ in $\op{P}(\set{U}\times\set{X})$, it is possible to achieve any value for $t=[p(b,1)-p(a,1)]/[p(a,0)-p(b,0)]$ in $(-\infty,+\infty)$. The theorem by Alberti and Uhlmann therefore guarantees the existence of a CPTP map between $\rho_i$ and $\sigma_i$.
\end{proof}

Lemma~\ref{lem:qubit} provides a neat second-law--like statement valid in the case of 2-by-2 density matrices:

\begin{proposition}[Fully quantum second-law--like statement for qubits]
	\label{prop:thermo-qubit}
	Given two states $\rho,\sigma\in\op{S}(\mathbb{C}^2)$, there exists an $\omega$-preserving transition from $\rho$ to $\sigma$ if and only if
	$\rho\succth\sigma$.
	\end{proposition}

Unfortunately, however, the above corollary cannot be generalized to higher dimensions, since counterexamples to the Alberti-Uhlmann theorem are known already when $\rho_i\in\op{S}(\mathbb{C}^3)$ and $\sigma_i\in\op{S}(\mathbb{C}^2)$~\cite{matsumoto_example_2014}. For the completely general case, we will need therefore a new ordering, stronger than that given in Definition~\ref{def:information-ordering} (for details, see Section~\ref{sec:full}).

In contrast, in the classical case, i.e., when $[\rho_0,\rho_1]=[\sigma_0,\sigma_1]=0$ (equivalently, $\rho_i$ and $\sigma_i$ are pairs of probability distributions $p_i\in\op{P}(\set{Z})$ and $q_i\in\op{P}(\set{Z}')$, respectively), a binary $\set{U}$ is \textit{always} sufficient. In jargon, this fact is stated saying that \textit{tests (2-decision problems) are sufficient for (classical) dichotomies}~\cite{blackwell_equivalent_1953}. Other results involve the comparisons using a complete set of $f$-divergences~\cite{torgersen1970comparison,torgersen_comparison_1991,liese_statistical_2008,matsumoto_condition_2014}. Anything more than a mention to some of the many deep results formulated in the classical scenario is far beyond the scope of the present notes.

\section{The semiquantum case: $[\sigma_0,\sigma_1]=0$}

A restricted, but still physically relevant scenario, is that in which the final state is known to commute the with Gibbs state---in thermodynamics, one would say that $\sigma$ is block-diagonal in the energy basis. Such a case is formalized by requiring that $[\sigma_0,\sigma_1]=0$. If this holds, we have the following:

\begin{lemma}\label{lem:semiquantum}
	Let $(\rho_0,\rho_1)$ and $(\sigma_0,\sigma_1)$ be two pairs of density matrices in $\op{S}(\sH_Q)$ and $\op{S}(\sH_{Q'})$, respectively. Assume moreover that $[\sigma_0,\sigma_1]=0$. Then, $(\rho_0,\rho_1)\succeq (\sigma_0,\sigma_1)$ if and only if there exists a CPTP map $\Phi:\op{L}(\sH_Q)\to\op{L}(\sH_{Q'})$ such that $\Phi(\rho_i)=\sigma_i$, for $i=0,1$. In fact, the comparison $(\rho_0,\rho_1)\succeq (\sigma_0,\sigma_1)$ can be restricted to sets $\set{U}$ with $|\set{U}|=\dim\sH_{Q'}$.
\end{lemma}

This fact, whose detailed proof will be given in a forthcoming paper~\cite{buscemi_forthcoming_1}, directly implies the following semiquantum second-law--like statement:

\begin{proposition}[Semiquantum second-law--like statement]
	\label{prop:thermo-semiquantum}
	Given two states $\rho\in\op{S}(\sH_Q)$ and $\sigma\in\op{S}(\sH_{Q'})$, with $[\sigma,\omega]=0$ (i.e., $\sigma$ is block-diagonal on the energy basis), there exists an $\omega$-preserving transition from $\rho$ to $\sigma$ if and only if
	$\rho\succth\sigma$.
\end{proposition}

\section{The fully general case}\label{sec:full}

In the general quantum case in which neither the $\rho_i$'s nor the $\sigma_i$'s are assumed to be commuting, Lemma~\ref{lem:semiquantum} above does not guarantee the existence of a quantum channel, but only the existence of a \textit{statistical morphism} (namely, a suitable generalization of the concept of positive CPTP linear maps~\cite{buscemi2012comparison,matsumoto_example_2014}).

We therefore need to introduce a new ordering, stronger (i.e., more stringent) than $\succeq$. In order to do this, we start with a definition:

\begin{definition}[Complete cq-channels] Given a Hilbert space $\sH$, a \emph{complete cq-channel on $\sH$} is defined as a classical-to-quantum channel $\mN:\set{Y}\to\op{S}(\sH)$, such that $|\set{Y}|=(\dim\sH)^2$, $y\mapsto\tau_y$, and $\operatorname{span}\{\tau_y:y\in\set{Y}\}=\op{L}(\sH)$.
\end{definition}

We then extend the construction of Section~\ref{sec:info-ordering} as follows. For any set $\set{U}$, any auxiliary Hilbert space $\sH_R$, any complete cq-channel $y\mapsto\tau^y_R$, and any joint probability distribution $p\in\op{P}(\set{U}\times\set{Y}\times\set{X})$, we define
\[
\begin{split}
\tilde{\rho}_{URQ}&\defeq\sum_{u\in\set{U}}\sum_{y\in\set{Y}}\sum_{x\in\set{X}}p(u,y,x)|u\>\<u|_U\otimes\tau^y_R\otimes\rho^x_Q\\
&=\sum_{u\in\set{U}}p(u)|u\>\<u|_U\otimes\tilde{\rho}^u_{RQ},
\end{split}
\]
with $\tilde{\rho}^u_{RQ}\defeq\sum_{y\in\set{Y}}\sum_{x\in\set{X}}p(y,x|u)\tau^y_R\otimes\rho^x_Q$. In the same way, we also define
\[
\begin{split}
\tilde{\sigma}_{URQ'}&\defeq\sum_{u\in\set{U}}\sum_{y\in\set{Y}}\sum_{x\in\set{X}}p(u,y,x)|u\>\<u|_U\otimes\tau^y_R\otimes\sigma^x_{Q'}\\
&=\sum_{u\in\set{U}}p(u)|u\>\<u|_U\otimes\tilde{\sigma}^u_{RQ'},
\end{split}
\]
with $\tilde{\sigma}^u_{RQ'}\defeq\sum_{y\in\set{Y}}\sum_{x\in\set{X}}p(y,x|u)\tau^y_R\otimes\sigma^x_{Q'}$.

What we want to do is to compare the information content of $(\rho_0,\rho_1)$ with that of $(\sigma_0,\sigma_1)$, not only in the case in which they are the only quantum resources available, but also in the case in which an additional quantum resource, represented by the family of states $\{t^y_R:y\in\set{Y}\}$, is (or could be, in principle) added to both.

\begin{definition}[Complete information ordering]
Let $(\rho_0,\rho_1)$ and $(\sigma_0,\sigma_1)$ be two pairs of density matrices in $\op{S}(\sH_Q)$ and in $\op{S}(\sH_{Q'})$, respectively. We say that $(\rho_0,\rho_1)$ is \emph{`completely more informative'} than $(\sigma_0,\sigma_1)$, and write \[(\rho_0,\rho_1)\succeq^* (\sigma_0,\sigma_1),\] if and only if, for any set $\set{U}$, any auxiliary Hilbert space $\sH_R$, any complete cq-channel $y\mapsto\tau^y_R$, and any joint probability distribution $p\in\op{P}(\set{U}\times\set{Y}\times\set{X})$
\[
\Hmin(U|RQ)_{\tilde{\rho}}\le\Hmin(U|RQ')_{\tilde{\sigma}}.
\]
In particular, we will write $\rho\succth^*\sigma$ if and only if $(\rho,\omega)\succeq^*(\sigma,\omega)$.
\end{definition}

We can then state the following:

\begin{lemma}\label{lem:quantum}
	Let $(\rho_0,\rho_1)$ and $(\sigma_0,\sigma_1)$ be two pairs of density matrices in $\op{S}(\sH_Q)$ and $\op{S}(\sH_{Q'})$, respectively. Then, $(\rho_0,\rho_1)\succeq^* (\sigma_0,\sigma_1)$ if and only if there exists a CPTP map $\Phi:\op{L}(\sH_Q)\to\op{L}(\sH_{Q'})$ such that $\Phi(\rho_i)=\sigma_i$, for $i=0,1$. In fact, the comparison $(\rho_0,\rho_1)\succeq^* (\sigma_0,\sigma_1)$ can be restricted to Hilbert spaces $\sH_R\cong\sH_{Q'}$ and sets $\set{U}$ with $|\set{U}|=(\dim\sH_{Q'})^2$.
\end{lemma}

This fact, whose detailed proof will be given in a forthcoming paper~\cite{buscemi_forthcoming_1}, directly implies the following second-law--like statement:

\begin{proposition}[Fully quantum second-law--like statement]
	\label{prop:thermo-quantum}
	Given two states $\rho\in\op{S}(\sH_Q)$ and $\sigma\in\op{S}(\sH_{Q'})$, there exists an $\omega$-preserving transition from $\rho$ to $\sigma$ if and only if
	$\rho\succth^*\sigma$.
\end{proposition}

\section{Concluding remarks}

\subsection{From partial orderings to total orderings}

The existence of a channel from one pair of states into another induces a partial ordering, in the sense that, in general, neither $\rho_i\mapsto\sigma_i$ nor viceversa. A way to introduce a \textit{total} ordering is to allow for an accuracy parameter $\epsilon\ge0$. This constitutes the main advantage of LeCam's theory of $\epsilon$-deficiencies with respect to Blackwell's theory of sufficiency. Indeed, in the classical scenario, a wealth of results are known for such an approximate scenario~\cite{cam1964sufficiency,cam_asymptotic_1986,torgersen_comparison_1991,liese_statistical_2008}, so that the problem can be considered, at least with respect to the main theoretical framework, complete. The same cannot be unfortunately said for the quantum case, despite some encouraging results in this direction~\cite{matsumoto2010quantum,jencova_comparison_2012}, mainly because the concepts of `quantum decisions,' needed in such a scenario, lacks of a clear-cut, direct operational interpretation. Further investigations are ongoing.

\subsection{Less noisy channels}

If, in Definition~\ref{def:information-ordering}, we replace $\Hmin$ with the standard (von Neumann) conditional entropy, then we have the definition of \textit{less noisy channels}~\cite{Korner1977,ElGamal1977}. Notice, however, that while the relation `less noisy' is known to be necessary, but not sufficient, for degradability~\cite{Korner1977}, the `one-shot' analogue proposed here is instead equivalent to degradability. This fact provides an interesting connection between statistical comparisons and information-theoretic orderings of communication channels, that will be the subject of a forthcoming contribution~\cite{buscemi_forthcoming_2} (along similar lines, see also Ref.~\cite{Raginsky2011}).

\appendix

\bibliography{C:/Users/Francesco/Desktop/library,C:/Users/Francesco/Desktop/zoterolibrary}

\end{document}